\documentclass[lettersize,journal]{IEEEtran}
\usepackage{amsmath,amsfonts,amssymb, amsthm}
\usepackage{algorithmic}
\usepackage{algorithm}
\usepackage{array}
\usepackage[caption=false,font=normalsize,labelfont=sf,textfont=sf]{subfig}
\usepackage{textcomp}
\usepackage{stfloats}
\usepackage{url}
\usepackage{verbatim}
\newtheorem{lemma}{Lemma} 
\usepackage{graphicx}
\usepackage{hyperref}
\usepackage{cleveref}
\usepackage{cite}
\usepackage{lettrine}
\usepackage[justification=centering]{caption}
\begin{document}

\title{Enhancing Channel Estimation for OTFS systems using Sparse Bayesian Learning with Adaptive Threshold}

\author{\IEEEauthorblockN{Tengfei Qi,
Yifei Yang, Xiong Deng\textsuperscript{*}, Zhinan Sun, Ziqiang Gao, Xihua Zou, Wei Pan, and Lianshan Yan\\}

 \textit{ School of Information Science and Technology, Southwest Jiaotong University}, Chengdu, China\\
 \textit{ National Mobile Communications Research Laboratory, Southeast University}, Nanjing, China\\
 *Corresponding author: xiongdeng@swjtu.edu.cn
}



\maketitle

\begin{abstract}
Orthogonal time frequency space (OTFS) modulation is a two-dimensional modulation scheme designed in the delay-Doppler (DD) domain, exhibiting superior performance over orthogonal frequency division multiplexing (OFDM) modulation in environments with high Doppler frequency shifts. We investigated the channel estimation in the DD domain of OTFS systems, modeling it as a sparse signal recovery problem. Subsequently, within the existing sparse Bayesian learning framework, we proposed an adaptive Bayesian threshold-based active denoising mechanism. Combined with inverse-free sparse Bayesian learning, this effectively addresses the pseudo-peak issue in low signal-to-noise ratio (SNR) scenarios while maintaining low complexity. The simulation results demonstrate that this algorithm outperforms existing channel estimation algorithms in terms of anti-noise performance and complexity.
\end{abstract}


\section{Introduction}
The sixth-generation wireless communication system are expected to support ubiquitous connectivity to a wide range of mobile devices, from self-driving cars to high-speed trains. One of the key challenges for these services is to provide reliable communications in a highly mobile environment\cite{Wei2021WC}. Conventional orthogonal frequency division multiplexing (OFDM) modulation is a key technology in fifth-generation (5G) cellular systems, which achieves high spectral efficiency for time-invariant frequency-selective channels. While OFDM systems are sensitive in high mobility scenarios and affected by the Doppler effect in time-varying channels. Recently, a new two-dimensional modulation scheme, i.e., orthogonal time frequency space (OTFS), has been proposed as a promising candidate for highly mobile communications \cite{Hadani2018}. OTFS modulates information in the delay-Doppler (DD) domain, rather than in the time-frequency (TF) domain of classical OFDM modulation, with sufficient diversity gain and significant anti-fading performance \cite{diversity2019TWC}.

The authors in \cite{Embedded2019TVT} employed a threshold-based method for channel estimation, which effectively reduces noise in the channel estimation but is only applicable in high signal-to-noise ratio (SNR) scenarios. In \cite{Murali2018ITA}, a channel estimation method based on pseudo-noise (PN) sequences was proposed, which estimated the fractional Doppler frequency shift of each channel through matching filtering. This method does not require high SNR, but has a large amount of computation, and has high requirements for the two-dimensional correlation of the sequence.The study in \cite{Embedded2019TVT} used the sparsity of channels in the DD domain to express the OTFS channel estimation as a sparse signal recovery problem, and the orthogonal matching pursuit (OMP) algorithm was applied to the channel estimation of OTFS. Sparse Bayesian learning (SBL) was used to solve the OTFS channel estimation problem in \cite{ZhaoLei2020cletter}, but the high complexity will hinder its practical application.

Recently, it has been found that by further relaxing the evidence lower bound, SBL requiring no inverse operation can be constructed\cite{IFSBL2017LSP}. This breakthrough significantly reduces computational complexity, achieving efficient solutions to sparse problems. However, this relaxation inevitably affects the sparsity and stability of the solution. Simulation results show that this relaxation strategy leads to unnecessary pseudo peaks at nonzero element positions, which not only reduces the accuracy of signal convergence but also may increase the bit error rate (BER). Especially in dealing with fractional Doppler effects, where channel responses typically exhibit block-sparse signals, the pseudo-peak effect caused by this algorithm is more pronounced. 

To address the issue of pseudo peaks generated in SBL algorithms, we propose an active denoising mechanism for OTFS system based on Bayesian threshold, which is cleverly embedded into the framework of SBL and continuously optimized through iteration. The core of this mechanism lies in its adaptive screening ability, which accurately identifies the positions of nonsparse values and actively adjusts the variance of these values. By actively adjusting the variance, our mechanism prevents nonsparse values from erroneously converging into pseudo peaks, thus ensuring the sparsity of the solution. 
Simulation analysis confirms that at low SNR, the proposed algorithm achieves a 0.6 dB improvement in channel estimation performance  with reduced complexity compared to the SBL algorithm. Compared to the inverse-free sparse Bayesian learning (IFSBL) algorithm, it maintains the same complexity and achieves a 2.6 dB improvement in channel estimation performance.


\section{OTFS System Model}
\subsection{OTFS Input and Output Model}
We consider a point-to-point single-antenna OTFS system. The information bit data to be sent is first modulated to form an information symbol, and \(x[k,l]\) is defined as the information symbol in the delay-Doppler domain, where \(M\) and \(N\) represent the number of subcarriers and the number of symbols, respectively. 
The signal \( x[k, l] \) is transformed into the time-frequency domain via inverse symplectic finite Fourier
transform (ISFFT) to obtain \( X[n, m] \). Subsequently, it is converted into the time domain using the Heisenberg transform to yield the time-domain signal \( s(t) \). This time-domain signal \( s(t) \) passes through a wireless channel characterized by a complex baseband channel impulse response \( h(\tau, \nu) \), where \( \tau \) and \( \nu \) represent the delay and Doppler shift of the time-varying channel, respectively. At the receiver, the received time-domain signal \( r(t) \) is first transformed into the time-frequency domain, and then discretely sampled to obtain \( Y[n, m] \).
Then, the SFFT transformation of \(Y[n,m]\) is performed to obtain the delay-Doppler domain symbol \(y[k,l]\), which is used for channel estimation and symbol detection\cite{OTFS2017WCNC}.

In a typical wireless environment, the maximum Doppler and delay shifts are limited. We denote the maximum Doppler and delay as \(\nu_{\max }\) and \(\tau_{\max }\), respectively, and the range of the DD domain channel \(h(\tau,\nu)\) is limited to the range of \([0,\tau_{\max }]\) and \([-\nu_{\max },\nu_{\max }]\). Depending on the sparsity of the DD domain channel, the channel can be modeled as
\begin{equation}
h\left(\nu,\tau\right)=\sum_{i=1}^{P}h_{i}\delta\left(\nu-\nu_{i}\right)\delta\left(\tau-\tau_{i}\right)
\end{equation}
where \(P\) is the number of propagation paths, \(h_i\), \(\tau_i\), and \(\nu_i\) denote the complex gain, delay and Doppler shifts of the \(i\)-th path, respectively, and \(\delta(\cdot)\) is the Dirac delta function. We represent the delay and Doppler taps for path \(i\)-th as follows
\begin{equation}
\tau_i=\frac{l_{\tau_i}}{M\Delta f},\nu_i=\frac{k_{\nu_i}+\kappa_{\nu_i}}{NT}
\end{equation}
where \(\Delta f\) and \(T\) are the sub-carrier spacing and symbol period, respectively. \(l_{\tau_i}\) and \(k_{\nu_i}\) are integers, and \(\kappa_{\nu_i}\) belongs to \((-0.5, 0.5)\). Specifically, \(l_{\tau_i}\) and \(k_{\nu_i}\) denote the indices of the delay taps and the Doppler taps, corresponding to the delay \(\tau_i\) and the Doppler frequency \(v_i\), respectively. Additionally, \(l_{\max }= \lceil \tau_{\max }M\Delta f \rceil \) and \(k_{\max } = \lceil \nu_{\max }NT \rceil\) are defined as the maximum delay and the maximum Doppler shift corresponding to the delay and Doppler taps. We refer to \(\kappa_{\nu_i}\) as fractional Doppler because it represents the fractional offset from the nearest Doppler tap \(k_{\nu_i}\). Fractional delay can be ignored because the resolution of the timeline is sufficient to approximate the path delay to the nearest sample point in a typical broadband system \cite{ZhaoLei2020cletter}. According to \cite{Interference2018TWC}, when both the transmit pulse \(g_{tx}\) and the receive pulse \(g_{rx}\) are rectangular pulse functions, the receive signal \(y[k,l]\) can be expressed as
\begin{align}
y[k,l]\approx&\sum_{i=1}^{P}\sum_{q=-\eta}^{\eta}(h_{i}e^{j2\pi\left(\frac{l-l_{\tau_{i}}}{M}\right)\left(\frac{k_{\nu_{i}}+\kappa_{\nu_{i}}}{N}\right)}\varphi_{i}(k,l,q) \notag
\\
&\cdot x\left[[k-k_{\nu_{i}}+q]_{N},[l-l_{\tau_{i}}]_{M}\right])+w[k,l]
\label{IOmodel}
\end{align}
where
\begin{align}
\varphi_i(k,l,q)=\begin{cases}\Psi(q) ~e^{j2\pi\frac{(l-l_{\tau_i})(k_{\nu_i}+\kappa_{\nu_i})}{MN}}\\~~~~~~l_{\tau_i}\leq l<M\\\left({\Psi(q)}-\frac{1}{N}\right)e^{j2\pi\frac{(l-l_{\tau_i})(k_{\nu_i}+\kappa_{\nu_i})}{MN}} e^{-j2\pi\frac{[k-k\nu_i+q]_N}{N}}\\~~~~~~0\leq l<l_{\tau_i}\end{cases}
\end{align}
and \(\Psi(q)=\frac{e^{j2\pi(-q-\kappa_{\nu_i})}-1}{Ne^{j\frac{2\pi}{N}(-q-\kappa_{\nu_i})}-N}\), \([\cdot]_M\) and \([\cdot]_N\)  is the modulo operation of \(M\) and \(N\), respectively. The role of \(\eta\) is to approximate the effect of fractional-order Doppler on the channel, and according to \cite{Embedded2019TVT}, the approximation works best when \(\eta\) = 5.

\subsection{OTFS Channel Estimation Model}
The arrangement of the pilot symbols, the protection symbols and the data symbols is shown in Fig. \ref{fig_0}. Define \(\mathcal{M}_p\) and  \(\mathcal{N}_p\) as the pilot index sets in the Doppler and delay domains, respectively. Pilot symbols of size \(|\mathcal{M}_p| \times |\mathcal{N}_p|\) are placed at the center of the DD domain, where \(|\cdot|\) returns the size of the input value. Following \cite{offgrid2022Wei}, placing null as the guard interval on the DD domain grids in the range: \(\min{\{\mathcal{N}_p}\} - 2{k_{\max }} - {\eta} \le k \le \max{\{\mathcal{N}_p}\} + 2{k_{\max }} + {\eta}\) and \(\min{\{\mathcal{M}_p}\} - {l_{\max }} \le l \le \max{\{\mathcal{M}_p}\} + {l_{\max }}\), with \(k \notin \mathcal{N}_p\) and \(l \notin \mathcal{M}_p\). The rest of the region can be used to place data symbols. 
As can be seen from Fig. \ref{fig_0}, the number of receiver symbols used for delay-Doppler channel estimation is \(Q = (2{k_{\max }} + |\mathcal{N}_p|)\times({l_{\max }} + |\mathcal{M}_p|)\), and we define \(R = (2{k_{\max }} + |\mathcal{N}_p| + 2{\eta})\times({l_{\max }} + |\mathcal{M}_p|)\) . To facilitate channel estimation, the system model in (1) can be rewritten as \cite{ZhaoLei2020cletter}
\begin{equation}
\mathbf y=\mathbf \Phi \mathbf h+ \mathbf w
\label{y_phih_w}
\end{equation}
where \(\mathbf y\) \(\in \) \(\mathbb{C}^{Q\times1}\),  \(\mathbf \Phi\) \(\in \) \(\mathbb{C}^{Q\times R}\), \(\mathbf h\) \(\in \) \(\mathbb{C}^{R\times1}\)  and \(\mathbf w\) \(\in \) \(\mathbb{C}^{Q\times1}\) are the receive signal, measurement matrix, channel and Gaussian noise, respectively.

\begin{figure}[!t]
\centering
\includegraphics[width=3in]{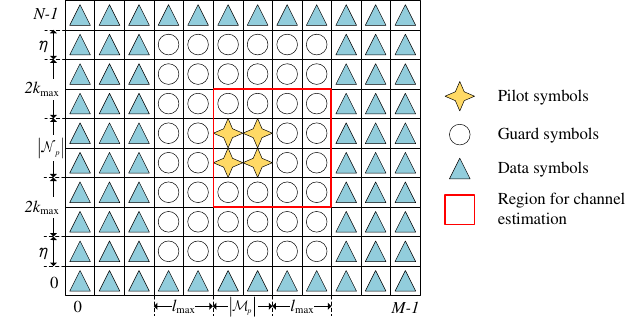}
\centering
\captionsetup{font={footnotesize }}
\caption{The arrangement of pilot and data symbols in the DD domain.}
\label{fig_0}
\end{figure}

\section{Channel Estimation Based On Embedded Bayesian Threshold}
In this section, we first review the basic process of IFSBL. The relaxation strategy used to avoid inverse operations will produce pseudo peaks at nonzero positions, affecting the accuracy of the final channel estimation. To address this issue, we propose an active denoising mechanism based on Bayesian thresholds that continuously adjusts and optimizes with each iteration.
\subsection{Inverse-free SBL for OTFS Channel Estimation}
The core idea of sparse Bayesian learning lies in leveraging prior information to impose sparsity constraints on model parameters and solving the posterior distribution of model parameters through Bayesian inference methods. In OTFS channel estimation, a two-layer hierarchical prior is assigned to \(\mathbf{h}\). As shown in \cite{IFSBL2017LSP}, \(\mathbf{h}\) follows a Gaussian prior distribution characterized by \(\boldsymbol{\alpha}\) in the first layer. The second layer imposes a Gamma hyper-prior on \(\boldsymbol{\alpha}\), parameterized by shape \(a\) and inverse scale \(b\). Similarly, \(\mathbf{w}\) is assumed to be Gaussian distribution with zero mean and covariance \((1/\gamma)\mathbf{I}\), while a Gamma hyper-prior is set for \(\gamma\) with shape \(c\) and scale \(d\). The Gamma distribution that \(\alpha_n\) and \(\gamma\) follow can be represented by \eqref{eq:alpha-gamma}. Let \(\boldsymbol{\theta} = \{\mathbf{h}, \boldsymbol{\alpha}, \gamma\}\) represent the latent variables for channel gain, variance of channel gain, and noise variance. Using mean-field theory, the approximated posterior distribution is expressed as \(q(\boldsymbol{\theta}) = q_h(\mathbf{h})q_\alpha(\boldsymbol{\alpha})q_\gamma(\gamma)\), as in \cite{tzikas2008MSP}.

\begin{equation}
p(\boldsymbol{\alpha})=\prod_{n=1}^Q\Gamma^{-1}(a)b^a\alpha_n^{a-1}e^{-b\alpha_n},
 ~~
 p(\gamma)=\Gamma(c)^{-1}d^c\gamma^{c-1}e^{-d\gamma} \label{eq:alpha-gamma} 
\end{equation}

In variational inference of SBL, the Kullback-Leibler (KL) divergence is commonly used to measure the similarity between the approximated posterior distribution \( q(\boldsymbol\theta) \) and the true posterior distributions \( p(\boldsymbol\theta|\mathbf y) \) with the following relationship
\begin{equation}
\text{ELBO}(q) = \ln p(\mathbf y)-\text{KL}(q || p) 
\end{equation} 
where \(\text{KL}(q || p)\) is the KL divergence between \( q(\boldsymbol\theta) \) and \( p(\boldsymbol\theta|\mathbf y) \).
Thus, minimizing the KL divergence is equivalent to maximizing the evidence lower bound (ELBO) by alternating optimization on  \( \mathbf{h},\boldsymbol\alpha\) and \(\gamma\) in \cite{IFSBL2017LSP}. Specifically, \(q_h(\mathbf{h}) \) is updated based on a Gaussian distribution, with its mean and covariance matrix defined by 
\begin{align}
    \boldsymbol{\mu}  =& ~\gamma \mathbf{\Sigma} {\mathbf \Phi^\text{T}}\mathbf{y},\\
    \mathbf{\Sigma}  =& {\left( {\gamma {\mathbf \Phi^\text{T}}\mathbf \Phi + \mathbf{\Lambda}} \right)^{ - 1}}
\end{align}
where \(\mathbf{\Lambda}\) is a diagonal matrix, and the \(n\)th diagonal element is equal to \(\left\langle {{\alpha _n}} \right\rangle \). Notice that updating the posterior distribution \(q_h(\mathbf{h})\) requires computing the inverse of matrix, suggesting maximizing a relaxed ELBO to tackle this computational challenge by Lipschitz constraints in \cite{IFSBL2017LSP}. Maximizing the ELBO requires \(V = 2{\lambda _{\max }}({\mathbf \Phi^\text{T}}\mathbf \Phi)\), where \({\lambda _{\max }}(\cdot)\) denotes finding the maximum eigenvalue of the input value. Then, we can obtain a relaxed lower bound on \(p(\mathbf{y}|\mathbf{h},\gamma )\) as
\begin{align}
        p(\mathbf{y}|\mathbf{h},\gamma ) \ge \frac{{\gamma ^{Q/2}}}{{\sqrt {2\pi } }}\exp \left( { - \frac{\gamma }{2}g(\mathbf{h},\mathbf{z})} \right)  \label{eq:ieq_g}
\end{align}
where
\begin{align}
    g(\mathbf{h},\mathbf{z}) \buildrel \Delta \over = ||\mathbf{y} - \mathbf{\Phi z}||_2^2 + 2{(\mathbf{h} - \mathbf{z})^\text{T}}{\mathbf \Phi^\text{T}}(\mathbf{\Phi z} - \mathbf{y}) + \frac{V}{2}||\mathbf{h} - \mathbf{z}||_2^2
\end{align}
and \(\mathbf{z}\) is the estimated vector.



\subsection{Expectation-Maximization Algorithm}
According to \cite{IFSBL2017LSP}, ignoring those terms that are independent of \(\mathbf{h}\), the approximate posterior distribution \({q_h}(\mathbf{h})\) can be calculated by
\begin{align}
    \ln {q_h}(\mathbf{h}) 
&\propto {} - {\mathbf{h}^\text{T}}\left( {\frac{{V\left\langle \gamma  \right\rangle }}{2}\mathbf{I} + \mathbf{\Lambda} } \right)\mathbf{\mathbf{h}} \notag
\\
&\phantom{\propto} +\left\langle \gamma  \right\rangle {\mathbf{h}^\text{T}}(2{\mathbf \Phi^\text{T}}(\mathbf{\Phi z} - \mathbf{y}) - V\mathbf{z}). \label{eq:q_x}
\end{align}

Clearly, \({q_h}(\mathbf{h})\) follows a Gaussian distribution with its mean and covariance matrix, respectively, given as
\begin{align}
\boldsymbol{\mu} & = -\left\langle \gamma  \right\rangle \mathbf{\Sigma} \left( {{\mathbf \Phi^\text{T}}\mathbf{\Phi z} - {\mathbf \Phi^\text{T}}\mathbf{y} - \frac{V}{2}\mathbf{z}} \right),
\label{eq:miu}
   \\ \mathbf{\Sigma}  &= {\left( {\frac{{V\left\langle \gamma  \right\rangle }}{2}\mathbf{I} + \mathbf{\mathbf{\Lambda}} } \right)^{ - 1}}
\label{eq:sigma}
\end{align}

In a similar manner, \(\boldsymbol\alpha \) can be expressed as a product of Gamma distributions and the parameters \({\tilde a}\) and \({{\tilde b}_n}\) are given by
\begin{align}
\tilde a = a + \frac{1}{2}{\kern 1pt} ,
 ~~
{{\tilde b}_n} = b + \frac{1}{2}\left\langle {h_n^2} \right\rangle. \label{eq:a_b_} 
\end{align}

Likewise, \(\gamma \) follows a Gamma distribution with the parameters \({\tilde c}\) and \({\tilde d}\) given as
\begin{align}
    \tilde c = c + \frac{R}{2} 
    ,~\tilde d = d + \frac{1}{2}\left\langle {g(\mathbf{h},\mathbf{z})} \right\rangle \label{eq:c_d_}.
\end{align}

In summary, the E-step involves update of the posterior approximations for hidden variables \(\mathbf{h}\), \(\boldsymbol\alpha \), and \(\gamma \). Some of the expectations and moments used during the update are summarized as
\begin{align}
    \left\langle {{\alpha _n}} \right\rangle  = \frac{{\tilde a}}{{{{\tilde b}_n}}}, 
    ~~ \left\langle \gamma  \right\rangle  = \frac{{\tilde c}}{{\tilde d}},
    ~~ \left\langle {h_n^2} \right\rangle  = \mu_n^2 + {\Sigma _{n,n}},
    \label{eq:expectation_updata}
\end{align}
\begin{align}
\left\langle {g(\mathbf{h},\mathbf{z})} \right\rangle  =& ||\mathbf{y} - \mathbf{\Phi z}||_2^2 + 2{(\boldsymbol{\mu}  - \mathbf{z})^\text{T}}{\mathbf \Phi^\text{T}}(\mathbf{\Phi z} - \mathbf{y})  
\notag
\\
&+ \frac{V}{2}(||\boldsymbol{\mu}  - \mathbf{z}||_2^2+ \text{Tr}(\mathbf{\Sigma} ))
\label{eq:g_update}
\end{align}
where \({\Sigma _{n,n}}\) represents the \(n\)-th diagonal element of \(\mathbf{\Sigma}\).

In M-step, we can use \(\mathbf{z} = \boldsymbol{\mu}\) to update the estimated vector to maximum the expectation according to \cite{IFSBL2017LSP}. Given the current estimate of \(\mathbf{z}\), update the posterior approximations \({q_h}(\mathbf{h}),{q_\alpha }(\boldsymbol\alpha ),{q_\gamma }(\gamma )\) according to \eqref{eq:miu}-\eqref{eq:sigma},\eqref{eq:a_b_},\eqref{eq:c_d_}.


\subsection{Embedded Bayesian Threshold }
The traditional Bayesian threshold mechanism has been successfully implemented in \cite{BHT_Zayyani2009ISP} and \cite{Kang2015TSP}. 
SBL has been extended to more practical scenarios in \cite{BHT2022TSP}. However, these algorithms focus mainly on the selection of sparse values and rely on fixed threshold settings. This may result in important sparse values being masked by noise and thus not effectively identified, particularly in low SNR. Therefore, it is imperative to develop a method capable of adaptively adjusting thresholds during the signal estimation process. 


Let \(H_1\) denote the hypothesis of nonzero elements when the SBL algorithm converges, while let \(H_0\) represent the hypothesis of zero elements. As depicted in \eqref{eq:test11}, when the probability of \(H_1\) is greater than \(H_0\), we consider this location to be a nonzero element. Otherwise, we consider this location as a zero element. 
\begin{equation}
H=
\left\{
\begin{array}{*{20}{c}}
H_1, & {P(H_1|\alpha_n)}>{P(H_0|\alpha_n)} , \\
H_0, & {P(H_1|\alpha_n)}<{P(H_0|\alpha_n)} .
\end{array}
\right.
\label{eq:test11}
\end{equation}

\begin{lemma}
    The Bayesian threshold for variance under the hypothesis when the SBL algorithm converges, can be written as the following iterative expression
\begin{align}
\rho_n^{(t + 1)} = \rho_n^{(t)} + \ln \frac{P(h_n^{(t)}|\alpha_n^{(t)},H_1)}{P(h_n^{(t)}|\alpha_n^{(t)},H_0)}.
\label{iter_BTH}
\end{align}
where \(\rho_n^{(t)} = \ln\frac{{P({H_1}|{\alpha _n^{(t)}})}}{{P({H_0}|{\alpha _n^{(t)}})}}\) is the result of the \(t\)-th iteration.
\end{lemma}

\begin{proof}
    Using the Bayesian formula, we can derive the expression for \({P({H_1}|{\alpha _n})}\) in  \eqref{eq:test11}
\begin{align}
P(H_1|\alpha_n) &= \frac{P(\alpha_n|H_1) P(H_1)}{P(\alpha_n)} \nonumber \\
&= \frac{P(H_1)}{P(\alpha_n)} \int P(\alpha_n|H_1,h_n)  P(h_n|H_1) dh_n\nonumber \\&
=P(H_1)  \int \frac{P(h_n|\alpha_n,H_1)}{P(h_n)} P(h_n|H_1)dh_n.\label{eq:test13}
\end{align}

Notably, there are certain differences between our stated assumptions and the general assumptions made in \cite{BHT2022TSP}. The difference is that our assumption is to judge zero and nonzero elements assuming that the current algorithm has converged while the general assumptions is to judge zero and nonzero elements after algorithm already converged. Since the assumption is made under algorithm convergence, the zero elements at this time will converge to \(\delta(h)\), that is, if and only if \(h=0\), the probability \(P({h}|{H_0})\) is 1. For non-zero elements, \(P({h}|{H_1})\) will converge to \(\delta(h-h_n)\). Although \(h_n\) is unknown, \(\delta(h-h_n)\) means that \(P({h}|{H_1})\) must converge with probability 1. Since \(h_n\) is estimated under the condition of convergence of the algorithm, but each round may be the final convergence result, it is reasonable to take each estimated result as the final convergence estimate result. Therefore, when incorporated into \eqref{eq:test13}, we can obtain
\begin{equation}
P(H_1|\alpha_n) = P(H_1)  \frac{P(h_n|\alpha_n,H_1)}{P(h_n)}.
\label{eq:test14}
\end{equation}

Similarly to \eqref{eq:test14}, \({P({H_0}|{\alpha _n})}\) can be transformed into the following expression through deduction
\begin{equation}
P(H_0|\alpha_n) = P(H_0)  \frac{P(h_n|\alpha_n,H_0)}{P(h_n)}.
\label{eq:test15}
\end{equation}

By dividing \eqref{eq:test14} and \eqref{eq:test15} and taking the logarithm, the following expression can be derived:
\begin{equation}
\ln\frac{P(H_1|\alpha_n)}{P(H_0|\alpha_n)} = \ln\frac{P(H_1)}{P(H_0)} + \ln\frac{P(h_n|\alpha_n,H_1)}{P(h_n|\alpha_n,H_0)}.
\label{eq:test16}
\end{equation}

In \eqref{eq:test16}, the left-hand side of the equation represents the posterior, the first term on the right-hand side represents the prior and the second term on the right-hand side represents the likelihood. We refrain from imposing specific constraints on the threshold to allow adaptive adjustment. 
Furthermore, we equate the posterior of the current iteration with the prior of the previous iteration. This iterative approach involves progressively approximating the posterior based on the previous posterior, serving as the updated prior. Consequently, this can be expressed as an iterative formulation
\begin{align}
\ln\frac{P^{(t + 1)}(H_1|\alpha_n^{( t)})}{P^{(t + 1)}(H_0|\alpha_n^{( t)})} =
\ln\frac{P^{( t)}(H_1|\alpha_n^{( t)})}{P^{\left( t \right)}(H_0|\alpha_n^{( t)})} + \ln\frac{P^{( t)}(h_n^{( t)}|\alpha_n^{( t)},H_1)}{P^{( t)}(h_n^{( t)}|\alpha_n^{( t)},H_0)}.
\label{eq:test17}
\end{align}

Let \(\rho_n^{(t)} = \ln\frac{{P({H_1}|{\alpha _n^{(t)}})}}{{P({H_0}|{\alpha _n^{(t)}})}}\), we can express \eqref{eq:test17} in the form of \eqref{iter_BTH}.


\end{proof}

In the actual iteration, since \(\boldsymbol{\mu} \), \(\boldsymbol\alpha\), \(\gamma\) are all known and dynamic in each iteration, using iterative likelihood estimation will also update the dynamic as a correction term. In \eqref{iter_BTH}, the final term is referred to as the likelihood function. We initialize the priors for \(H_0\) and \(H_1\) with equal probabilities at the start of iteration. During each iteration, we update the likelihood function and subsequently use the obtained posterior as the prior for the next iteration.

\subsection{Proposed Algorithm}
The main idea on denoising is to fix the noise variance selected by Bayesian threshold during the SBL operation. Therefore, we increase the variance of the noise by multiplying a constant \(\varsigma\) greater than 1, which will make the noise impossible to converge to a peak in \textbf{Algorithm} \ref{alg:IFSBL-T}. Since the prior on signals and noise follow a Gaussian distribution, which means \(h_n\) follows \(\frac{\sqrt{\alpha_n}}{\sqrt{2\pi}} \exp{(-\frac{\alpha_n h_n^2}{2} )}\) while \(w_n\) follows \(\frac{\sqrt{\gamma}}{\sqrt{2\pi}} \exp{(-\frac{\gamma h_n^2}{2} )}\). We can substitute the prior into the second term on the right side of equation \eqref{iter_BTH} to get the likelihood function of embedded Bayesian threshold for improving channel estimation of OTFS, which is 
\begin{align}
    \ln \frac{P(h_n|\alpha_n,H_1)}{P(h_n|\alpha_n,H_0)}
    =
    \frac{1}{2}\ln{\frac{\alpha_n}{\gamma}}
    -\frac{h_n^2}{2}(\alpha_n-\gamma)
    \label{EBTH-IFSBL-T}
\end{align}
where \(H_1\) is the non-zero element hypothesis and \(H_0\) is the zero element hypothesis.


Therefore, we can embed the proposed Bayesian threshold into the iterative operation of the IFSBL. The Bayesian threshold algorithm is detailed in \textbf{Algorithm \ref{alg:IFSBL-T}}.

\begin{algorithm}[H]
\caption{IFSBL with Embedded Bayesian Threshold}\label{alg:IFSBL-T}
\begin{algorithmic}[1]
\STATE {\textbf{Initialization:}}
\STATE \hspace{0.0cm}Set $\boldsymbol{\rho}^{(0
)} = \text{0}$, $\boldsymbol{\lambda} = \text{0}$, $a = b = c = d = 10^{-5}$, $iteration_{\max} = 10^3$,  $\epsilon = 10^{-8}$, $\varsigma=10$
\STATE {\textbf{While}}$(i<iteration_{\max} \ \textbf{and} \ \frac{\left \| \mathbf{z}^{(i)}-\mathbf{z}^{(i-1)} \right \|_2^2}{\left \| \mathbf{z}^{(i-1)} \right \|_2^2} \ge \epsilon ):$
\STATE \hspace{0.5cm}$  a,b,c,d\gets \eqref{eq:a_b_},\eqref{eq:c_d_} $
\STATE \hspace{0.5cm}$  \left\langle {{\alpha _n}} \right\rangle, \left\langle \gamma  \right\rangle,\left\langle {h_n^2} \right\rangle,\left\langle {g(\mathbf{h},\mathbf{z})} \right\rangle \gets \eqref{eq:expectation_updata},\eqref{eq:g_update}$
\STATE \hspace{0.5cm}$\boldsymbol{\rho}^{(i)} \gets \eqref{iter_BTH}$
\STATE \hspace{0.5cm}$ \{\xi\} \leftarrow \text{find position}  (\boldsymbol{\rho}^{(i)} <0)  $
\STATE \hspace{0.5cm}$\left \langle \alpha _{\xi } \right \rangle  \leftarrow \varsigma\cdot\left \langle \alpha _{\xi } \right \rangle  $
\STATE \hspace{0.5cm}$ \mathbf{\Sigma},\boldsymbol{\mu} \leftarrow \eqref{eq:sigma}, \eqref{eq:miu} $
\STATE \hspace{0.5cm}$\mathbf{z}^{(i)} \leftarrow \boldsymbol{\mu} $
\STATE \textbf{Return}\hspace{0.1cm}$\mathbf{\hat{h}}=\mathbf{z}^{(i)}$
\end{algorithmic}
\label{alg1}
\end{algorithm}

\section{Simulation Results} 
In this section, we compare the proposed algorithm with SBL\cite{tzikas2008MSP}, IFSBL\cite{IFSBL2017LSP}, and OMP\cite{OMP2007trans} algorithms. For simplicity, we refer to \textbf{Algorithm \ref{alg:IFSBL-T}} as IFSBL-T. We analyze the channel estimation performance of these algorithms using normalized mean square error (NMSE) and bit error rate (BER). The parameters of the OTFS system are set as follows: \(M\) = 128, \(N\) = 128, \(P\) = 9, carrier frequency = 4 GHz, \(\Delta f\) = 15 kHz, 4-QAM modulation, maximum delay tap \(l_\tau\) = 20 and Doppler tap \(k_\nu\) = 16, in which Doppler tap corresponds to user equipment speeds of 500 Km/h. We use the fast time-varying channel according to the extended vehicular A (EVA) channel model proposed by 3GPP.  NMSE is defined as \(10{\log _{10}}\left( {\frac{{\parallel \mathbf{h} - \mathbf{\hat h}\parallel _2^2}}{{\parallel \mathbf{h}\parallel _2^2}}} \right)\), where $\mathbf{h}$ is the true value of the channel state information (CSI) and $\mathbf{\hat h}$ is the approximate value.

\begin{figure}[!h]
\centering
\includegraphics[width=2.5in]{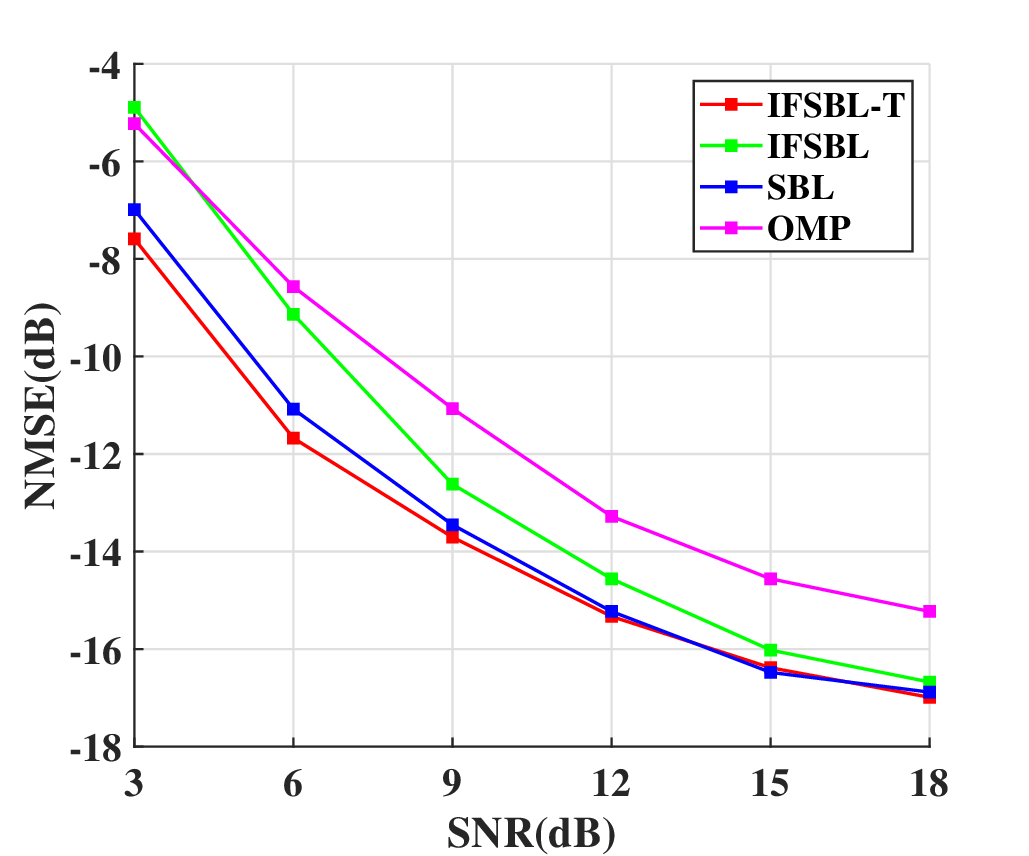}
\centering
\captionsetup{font={footnotesize }}
\caption{NMSE of OTFS channel estimation under different algorithms.}
\label{fig_1}
\end{figure}

We conducted simulations to assess the performance of various channel recovery algorithms in a SNR range of 3 to 18 dB. The results of the estimation of the OTFS channel are depicted in Fig. \ref{fig_1}. At high SNR values, the IFSBL-T algorithm and the SBL algorithm exhibit comparable performance. However, at low SNR values, the IFSBL-T algorithm outperforms both the IFSBL and SBL algorithms. Specifically, the IFSBL-T algorithm improves performance by approximately 2.6 dB compared to the IFSBL algorithm when SNR = 3 dB. This improved performance is attributed to the IFSBL-T algorithm's threshold-based selection of noise information, which effectively suppresses noise convergence into pseudo peaks through active assignment.

\begin{figure}[!h]
\centering
\includegraphics[width=2.5in]{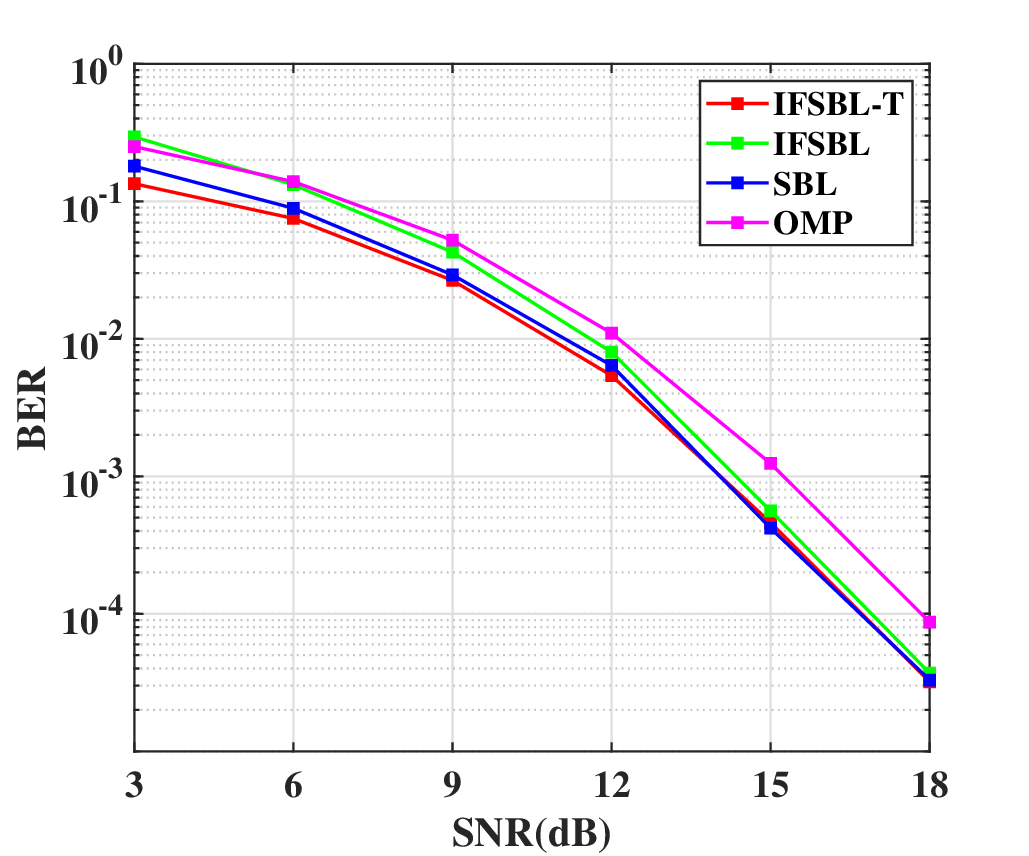}
\centering
\captionsetup{font={footnotesize }}
\caption{BER of OTFS under different algorithms.}
\label{fig_2}
\end{figure}

After acquiring the CSI, we employed the Message Passing algorithm for signal detection, with the simulation results depicted in Fig. \ref{fig_2}. It can be seen that the performance of IFSBL-T algorithm is generally better than that of traditional algorithms because IFSBL-T strictly inhibits the generation of pseudo peaks in places with non-sparse values. Since SBL assumes signals follow a Gaussian distribution, increasing the variance to suppress pseudo peaks can reduce noise in sparse signals to some extent, thereby improving BER performance. Fig. \ref{fig_1} and Fig. \ref{fig_2} demonstrate that the SBL performs similarly to the IFSBL-T  algorithm under high SNR conditions. However, under low SNR conditions, the performance of SBL is consistently inferior to that of the IFSBL-T algorithm due to the impact of noise.


\begin{figure}[!t]
\centering
\includegraphics[width=2.5in]{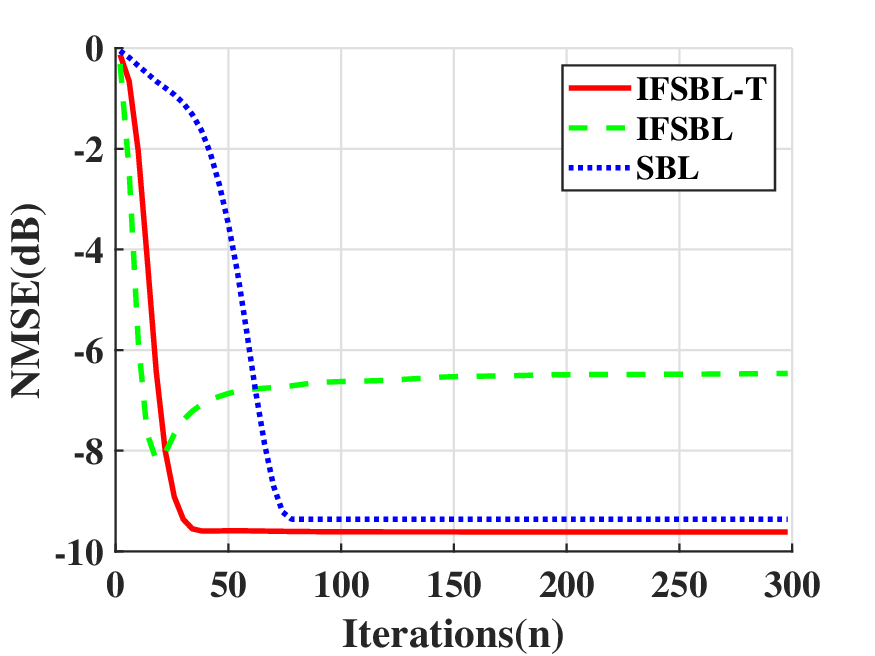}
\centering
\captionsetup{font={footnotesize }}
\caption{Comparison of convergence speed and performance among different algorithms.}
\label{fig_3}
\end{figure}
Furthermore, to more intuitively compare the performance differences in complexity, convergence speed, and other aspects among SBL, IFSBL, and IFSBL-T, we conducted two additional simulations. We set the SNR to 10 dB, the channel length to be reconstruct at 240, the received signal length at 180, and the number of propagation paths at 12. Using different SBL algorithms for 100 independent simulations, the results shown in Fig. \ref{fig_3} reveal significant performance differences. The IFSBL method requires over 100 iterations to achieve convergence. In contrast, SBL converges after approximately 80 iterations, while IFSBL-T achieves convergence in just 40 iterations. Furthermore, the IFSBL-T algorithm demonstrates a significantly lower NMSE compared to the IFSBL algorithm after reaching convergence.

\begin{figure}[!t]
\centering
\includegraphics[width=2.5in]{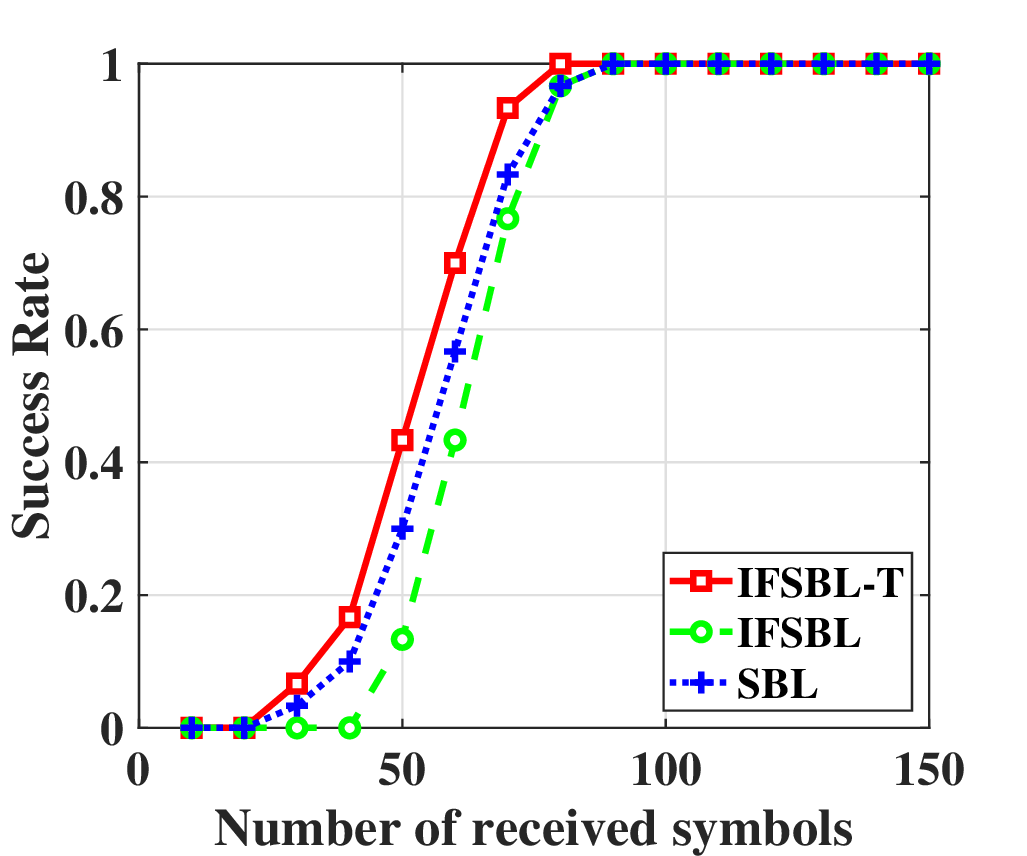}
\centering
\captionsetup{font={footnotesize }}
\caption{Success rates of respective algorithms versus Number of Measurement.}
\label{fig_4}
\end{figure}

Additionally, we conducted simulations comparing recovery success rates \cite{successratio2016AOS}, with results shown in Fig. \ref{fig_4}. This simulation demonstrates that the proposed algorithm can recover the original information with the fewest received symbols.

Regarding the computational complexity of the proposed algorithm, it is worth noting that the embedding of the Bayesian threshold within the iterative process necessitates a calculation for each iteration. These calculations are limited to element-level operations, rather than involving matrix computations. Consequently, the computational complexity of the overall algorithm remains consistent with that of IFSBL. In terms of space complexity, it is essential to consider the additional computing space required to store the Bayesian threshold obtained from the previous iteration. In particular, this additional space is equivalent to the length of the signal being processed. 

\section{Conclusion}
In this paper, we propose a channel estimation algorithm suitable for the OTFS system. To eliminate pseudo peaks in the IFSBL process under low SNR conditions, we develop an active denoising mechanism based on adaptive Bayesian thresholding. This approach effectively reduces noise, enhances the accuracy of channel estimation, and retains low algorithmic complexity. 
Simulation analysis at low SNR demonstrates that the proposed IFSBL-T algorithm reduces complexity and achieves a 0.6 dB improvement in channel estimation performance compared to the SBL algorithm. Furthermore, it shows a 2.6 dB improvement in channel estimation performance over the IFSBL algorithm while maintaining the same complexity.

\thispagestyle{empty}
\section*{Acknowledgment}
This work was supported by China National Key R\&D Programmes under Grant 2021YFB2800801, National Natural Science Foundation of China under Grant 62001174, 62271422, U23A20376, Sichuan Outstanding Youth Science and Technology Talents Project under Grant 2022JDJQ0047, Sichuan Science and Technology Program 2022ZYD0119.

\bibliographystyle{IEEEtran}
\bibliography{refer}

\end{document}